\documentclass{article}

\usepackage[preprint]{neurips_2026}

\usepackage{enumitem}
\usepackage{natbib}
\usepackage{graphicx,wrapfig}
\usepackage[caption = false]{subfig}
\usepackage{subcaption}
\usepackage{dcolumn}
\usepackage[utf8]{inputenc}
\usepackage{blindtext}
\usepackage{xcolor}
\usepackage{algorithm} 
\usepackage{algpseudocode} 
\usepackage{tikz}
\usetikzlibrary{trees}
 \usepackage{soul}
 \usepackage{bm}
 \usepackage{mathtools,amsmath,amssymb,amsthm}
 \usepackage{float}
 \usepackage[caption = false]{subfig}
 \usepackage{graphicx}
\usepackage{multirow}
\usetikzlibrary{patterns,shapes.geometric, arrows.meta,positioning}

\theoremstyle{definition}

\newtheorem{proposition}{Proposition}
\newtheorem{remark}{Remark}

\title{Beyond DSA: Conjugacy-based Comparison of Dynamical Systems}

\author{%
  Prakhar Godara \\
  Department of Psychology\\
  New York University, NY 10003, USA \\
  \texttt{prakhargodara@gmail.com} \\
   \And
   Pang Shiang Tay \\
  Department of Mathematics\\
  G\"ottingen University, 37073, Germany \\
  \texttt{pangshiang.tay@stud.uni-goettingen.de} \\
   \And
   Marcelo G Mattar \\
  Department of Psychology\\
  New York University, NY 10003, USA \\
  \texttt{mm13100@nyu.edu} 
}

\begin{document}

\maketitle

\begin{abstract}
Comparing whether two dynamical systems implement the same computation despite differences in coordinates or measurements is a central problem in neuroscience and machine learning. Dynamical Similarity Analysis \citep[DSA;][]{ostrow2023beyond} addresses this problem by aligning finite-dimensional Koopman approximations of the two systems through an orthogonal similarity transformation. Here we show that orthogonal alignment is neither necessary nor sufficient for topological conjugacy: genuinely conjugate systems may be related by a non-orthogonal basis-transfer matrix that DSA cannot capture, while non-conjugate systems may have orthogonally equivalent Koopman operators that DSA will fail to distinguish. We then use this observation to formulate \emph{Conjugacy-based Similarity Analysis} (CSA), which restricts alignments to those induced by candidate state-space bijections rather than arbitrary orthogonal matrices. We prove that CSA's fitted alignment is the finite-data projection of the composition operator associated with the candidate bijection, and use controlled examples to show why this distinction matters when observable dictionaries are chosen explicitly or implicitly from data. Together, these results clarify what Koopman-based similarity measures must ensure to support claims of identifying conjugacies between computational systems.
\end{abstract}

\section{Introduction}
\label{sec:introduction}

Comparing the computations performed by two dynamical systems is a central methodological question across neuroscience and machine learning. In neuroscience, such comparisons are used to validate computational models against neural recordings \citep{schrimpf2020brain, pagan2022new}, to align brain--computer interfaces across recording sessions \citep{degenhart2020stabilization, karpowicz2025stabilizing}, and to ask whether different circuits implement equivalent computations \citep{chaudhuri2019intrinsic}. In machine learning, similarity measures between learned representations are used to characterize training dynamics, architectural invariances, and model families \citep{kriegeskorte2008representational, raghu2017svcca, williams2021generalized}. Across both fields, the underlying question is the same: when do two systems compute in the same way? 

A natural way to make the question of computational equivalence precise is to ask whether one system can be obtained from the other by a relabelling of states. Concretely, suppose a bijection $h$ matches each state $x$ of one system to a state $h(x)$ of the other. 
If running the first system for one step and then relabelling produces the same result as relabelling first and then running the second system --- that is, if $h \circ f = g \circ h$---, then the two systems generate the same trajectories up to a renaming of their state variables. In dynamical systems theory this property is called \emph{topological conjugacy}. %Formally, two systems \(f:X\to X\) and \(g:Y\to Y\) are conjugate if there exists a bijection \(h:X\to Y\) rellabeling each state $x$ into $h(x)$, such that $h \circ f = g \circ h$. 
Conjugacy thus formalizes the intuition behind ``same computation'' by ignoring the particular coordinates in which states are described and asking only whether the underlying transition rule is the same. Geometric similarity, by contrast, fails to capture this invariance: it compares the shapes traced out by trajectories rather than the rules that generate them \citep{kriegeskorte2008representational, raghu2017svcca, williams2021generalized}.

Testing conjugacy directly is difficult because the relabelling function $h$ is unknown and lives in an infinite-dimensional space of candidate maps. Koopman operator theory offers a route around this. The Koopman operator represents a possibly nonlinear dynamical system through its linear action on observables, providing a principled language for learning and comparing dynamical structure from data \citep{koopman1931hamiltonian}. Building on this framework, Dynamical Similarity Analysis (DSA) was recently introduced to compare dynamical systems in terms of their finite-dimensional Koopman approximations rather than raw state-space geometry \citep{ostrow2023beyond}. DSA has since been adopted in systems neuroscience and deep learning as the operator-level alternative to geometric similarity.

DSA, however, looks for the wrong kind of match between Koopman operators. Conjugacy implies that a relabelling induces a specific transformation of measurements, and it is this transformation that should align the two Koopman operators. DSA ignores this, searching instead for any orthogonal alignment between Koopman matrices, whether or not it corresponds to a relabelling of states. This mismatch creates two failure modes: DSA can make genuinely different systems look similar when an arbitrary orthogonal alignment hides their difference, and can make genuinely equivalent systems look different when the correct state-by-state correspondence is outside its alignment class. In other words, DSA compares learned operator representations without requiring the alignment to be a valid change of coordinates in the underlying dynamical system. Consequently, neither a small nor a large DSA distance supports a clean conclusion about the dynamics: a small distance need not imply the same computation, and a large distance need not imply different computations.

To address this problem, we introduce \emph{conjugacy-based similarity analysis} (CSA), a template for comparing finite-dimensional Koopman approximations through candidate conjugacies rather than arbitrary orthogonal alignments. Given a parameterized family of candidate relabellings $\{h_\theta\}_{\theta \in \Theta}$ (eg: monotone  polynomials of a bounded degree, etc.), CSA constructs from each $h_\theta$ the matrix $P_\theta$ that records how measurements transform under $h_\theta$, and then asks how well this $P_\theta$ aligns the two Koopman matrices. The score is the smallest alignment error obtained as $\theta$ varies over the family. Because every $P_\theta$ comes, by construction, from an actual relabeling of states, any alignment found by CSA could have been produced by a conjugacy, a guarantee that does not hold for alignments found by DSA. Empirically, CSA correctly identifies conjugate systems in the very regimes where DSA's failure modes predict that it cannot, both when the measurement dictionary is chosen explicitly and implicitly from data.

Our contributions are threefold. \textbf{First}, we show that orthogonal alignment of Koopman matrices is not, in general, a valid criterion for conjugacy, even before approximation error is considered (Section~\ref{sec:diagnosis}). \textbf{Second}, we introduce CSA (Section~\ref{sec:CSA}) as a template for conjugacy-aware Koopman comparison, and prove that the matrix it fits is the correct finite-dimensional shadow of the composition operator induced by a candidate relabelling, so that CSA estimates a well-defined operator-theoretic object rather than a heuristic alignment. \textbf{Third}, we instantiate this template in controlled settings (Section~\ref{sec:results}) and show empirically that the distinction between orthogonal alignment and relabelling-induced alignment changes the conclusions one draws about conjugacy, both under state-space transformations and under changes of measurement basis.

\section{Theory}
\label{sec:background}

\subsection{Dynamical systems and conjugacy}
\label{sec:bg-conjugacy}

We consider discrete-time dynamical systems \(f:X\to X\) and \(g:Y\to Y\), where \(X\) and \(Y\) are state spaces. Two such systems are \emph{topologically conjugate} if there exists a bijection $h: X \to Y$ such that
\begin{equation}
\label{eq:spatial_conjugacy}
    h\circ f = g\circ h ,
\end{equation}
in which case $h$ provides a state-by-state correspondence under which the two systems generate the same trajectories. Testing conjugacy directly is difficult because the relabelling \(h\) is unknown and, in general, must be searched for over a large function space of candidate maps. The remainder of this section introduces the linear operators that allow us to work with conjugacy through measurements rather than through $h$ directly.

\subsection{Koopman operators}
\label{sec:bg-koopman}

The Koopman operator lifts a possibly nonlinear dynamical system to a linear operator on observables \citep{koopman1931hamiltonian}. Let \(\mathcal H_X\) and \(\mathcal H_Y\) be Hilbert spaces of scalar-valued observables on \(X\) and \(Y\). The Koopman operators associated with \(f\) and \(g\) are
\[
    (\mathcal K_f u)(x)=u(f(x)), \qquad
    (\mathcal K_g v)(y)=v(g(y)),
\]
for observables \(u\in\mathcal H_X\) and \(v\in\mathcal H_Y\). Thus, although \(f\) and \(g\) may be nonlinear maps on state space, \(\mathcal K_f\) and \(\mathcal K_g\) are linear operators on observable spaces. This replaces a nonlinear state-space comparison problem with a linear operator comparison problem, at the cost of working in infinite-dimensional function spaces.

\subsection{Composition operators induced by bijections}
\label{sec:bg-composition}

Any bijection between state spaces induces a corresponding linear map between observable spaces. Specifically, a bijection \(h: X \to Y\) induces the \emph{composition operator}
\begin{equation}
    \mathcal C_h:\mathcal H_Y\to\mathcal H_X,
    \qquad
    (\mathcal C_h v)(x)=v(h(x)).
\end{equation}
which pulls back observables on \(Y\) to observables on \(X\) along \(h\). The composition operator is the natural counterpart, at the level of measurements, of the relabelling \(h\) at the level of states. When \(f\) and \(g\) are conjugate, this specific operator---not an arbitrary linear map between observable spaces---connects \(\mathcal K_f\) and \(\mathcal K_g\). Indeed, the conjugacy relation in Eq.~\eqref{eq:spatial_conjugacy} implies that this composition operator intertwines the two Koopman operators. Thus, for any \(v\in\mathcal H_Y\),
\begin{align}
    (\mathcal C_h\mathcal K_g v)(x)
    = (\mathcal K_g v)(h(x)) 
    = v(g(h(x))) 
    = v(h(f(x))) 
    = (\mathcal C_h v)(f(x)) 
    = (\mathcal K_f\mathcal C_h v)(x).
\end{align}
Therefore
\begin{equation}
\label{eq:koopman_intertwining}
    \mathcal C_h\mathcal K_g
    =
    \mathcal K_f\mathcal C_h .
\end{equation}
This is the operator-theoretic consequence of topological conjugacy: the intertwiner is not an arbitrary operator on observables, but specifically the composition operator induced by the same state-space map \(h\) that conjugates the dynamics.

\subsection{Matrix representation of the intertwining relation}
\label{sec:bg-matrix}

To express the operator intertwining \eqref{eq:koopman_intertwining} in coordinates, let \(\{\phi_i\}_{i\geq 1}\) and \(\{\psi_j\}_{j\geq 1}\) be orthonormal bases of \(\mathcal H_X\) and \(\mathcal H_Y\), respectively, and write
\[
    \Phi(x)
    =
    (\phi_1(x),\phi_2(x),\ldots)^\top,
    \qquad
    \Psi(y)
    =
    (\psi_1(y),\psi_2(y),\ldots)^\top .
\]
The matrix representations $F$ and $G$ of the Koopman operators $\mathcal K_f$ and $\mathcal K_g$ in these bases are defined by
\begin{equation}
\label{eq:koopman_matrix_representations}
    \Phi(f(x)) = F\Phi(x),
    \qquad
    \Psi(g(y)) = G\Psi(y),
\end{equation}
and the matrix representation of the composition operator $\mathcal C_h$ is the matrix $P$ satisfying
\begin{equation}
\label{eq:composition_matrix_condition}
    \Psi(h(x)) = P\Phi(x).
\end{equation}
The $j$th row of $P$ gives the coefficients, in the $\Phi$-basis, of the pulled-back observable $\mathcal C_h \psi_j = \psi_j \circ h$; thus $P$ is not an arbitrary change of coordinates but the coordinate representation of a pullback. Combining Eqs.~\eqref{eq:spatial_conjugacy}, \eqref{eq:koopman_matrix_representations} and \eqref{eq:composition_matrix_condition} yields the matrix-level counterpart of the operator intertwining \eqref{eq:koopman_intertwining}:
\begin{align}
    G P \Phi(x)
    = G\Psi(h(x))
    = \Psi(g(h(x)))
    = \Psi(h(f(x)))
    = P\Phi(f(x))
    = P F\Phi(x).
\end{align}
Since this holds for all $x$, the matrices satisfy
\begin{equation}
\label{eq:matrix_intertwining}
    P F = G P,
\end{equation}
or equivalently, when $P$ is invertible,
\begin{equation}
\label{eq:matrix_similarity}
    G = P F P^{-1}.
\end{equation}
This is the matrix-level signature of conjugacy: the two Koopman representations are related not by an arbitrary alignment, but by the matrix representation of the composition operator induced by the state-space relabelling \(h\).

The next section asks whether the class of alignments used by DSA matches the class of alignments implied by conjugacy.

\section{Why DSA fails as a test of conjugacy}
\label{sec:diagnosis}

DSA attacks the conjugacy problem by directly aiming to approximate Eq. \eqref{eq:matrix_similarity}. However, as aforementioned, the matrix relation \(PF = GP\) is not an arbitrary similarity: \(P\) is the coordinate representation of a pullback induced by a candidate state-space map \(h\), and must therefore satisfy Eq. \eqref{eq:composition_matrix_condition} for some $h$---i.e., it cannot be chosen freely. Any criterion intended to detect conjugacy must therefore search over alignments that could plausibly arise from such pullbacks. The following two subsections show that DSA's orthogonal search fails this requirement, both in the alignment class it considers (\S\ref{sec:diagnosis-dsa-mismatch}) and in the information it discards when computing distances (\S\ref{sec:diagnosis-representational}).

\subsection{DSA searches over the wrong alignment class}
\label{sec:diagnosis-dsa-mismatch}

Given finite-dimensional Koopman approximations \(F\) and \(G\) of the same size, DSA computes
\begin{equation}
\label{eq:dsa_distance}
    d_{\mathrm{DSA}}(F,G)
    =
    \min_{Q\in O(m)}
    \|G-QFQ^{-1}\|_F .
\end{equation}
That is, DSA asks whether the two Koopman matrices can be made similar by an orthogonal change of coordinates, whereas \(PF=GP\) asks for a composition-induced \(P\). These two classes are not in containment in either direction.

On one hand, DSA is too permissive because not every orthogonal alignment is composition-induced. A defining property of a composition operator is that it is distributive over multiplication of observable functions, but a generic orthonormal operator does not have this property (Appendix~\ref{app:unitary_not_composition} provides a simple example). Crucially, this is not an artifact coming from the finite-dimensional approximations: classical results in ergodic theory exhibit measure-preserving systems whose true (infinite-dimensional) Koopman operators are unitarily equivalent without being topologically conjugate \citep{redei2012history}, and Appendix~\ref{app:counterexample} works through one such case for Bernoulli shifts. 

On the other hand, DSA is also too restrictive because not every composition-induced alignment is orthogonal. The composition operator \(\mathcal C_h\) is unitary only under additional assumptions, most notably when \(h\) is measure preserving \citep{budivsic2012applied} (see Appendix~\ref{app:when_is_P_unitary?}). Outside that special case, the correct intertwiner falls outside the orthogonal class searched by DSA, and DSA can report a large discrepancy even between genuinely conjugate systems.

\subsection{The mismatch is representational, not merely numerical}
\label{sec:diagnosis-representational}

Beyond optimizing over an inappropriate domain, DSA ignores too much information when computing distances. For instance, when solving the optimization problem in Eq. \eqref{eq:dsa_distance}, DSA ignores which subspaces are $F,G$ operators in and how those spaces are relatively oriented. Rather it only views them as numerical arrays and looks for similarity. The form this issue takes depends on how the Koopman approximation is learned. In EDMD, the observable dictionary is chosen explicitly, so the basis-transfer problem is directly visible. In Hankel-DMD one identifies the right subspace by taking the SVD of the Hankel matrix \citep{arbabi2017ergodic}. When comparing the corresponding Koopman matrices, DSA ignores the information of \textit{which} lags were chosen by the SVD procedure, which is arguably relevant for evaluating conjugacy.

These observations identify the design requirement for a conjugacy-aware comparison criterion. Rather than comparing Koopman matrices as bare numerical arrays, one should compare them through alignments that can be interpreted as projected pullbacks of candidate state-space maps. The next section introduces \emph{conjugacy-based similarity analysis} (CSA), which implements this idea: CSA first constructs the finite-dimensional transfer induced by a candidate relabelling, and only then asks whether that transfer intertwines the two Koopman approximations.

\section{Conjugacy-based Similarity Analysis (CSA)}
\label{sec:CSA}

%\paragraph{Setup.} 
Let \(f:X\to X\) and \(g:Y\to Y\) be discrete-time dynamical systems as above. We now assume finite observable dictionaries
\[
    \Phi=(\phi_1,\ldots,\phi_m)^\top:X\to\mathbb R^m,
    \qquad
    \Psi=(\psi_1,\ldots,\psi_n)^\top:Y\to\mathbb R^n,
\]
and finite-dimensional Koopman approximations
\(F\in\mathbb R^{m\times m}\) and \(G\in\mathbb R^{n\times n}\) (where $m$ need not be equal to $n$), for instance fitted by EDMD \citep{williams2015data} or Hankel-DMD \citep{arbabi2017ergodic}, satisfying
\[
    \Phi(f(x))\approx F\Phi(x),
    \qquad
    \Psi(g(y))\approx G\Psi(y).
\]

%\paragraph{Inner problem: the basis-transfer matrix.} 
CSA restricts the admissible alignments to those that arise from candidate state-space maps. Let \(\{h_\theta:X\to Y:\theta\in\Theta\}\) be a parameterized family of candidate bijections. This family specifies the class of conjugacies over which we search. In low-dimensional examples, it can be chosen explicitly; in higher-dimensional settings, one could use flexible parameterizations of diffeomorphisms, such as invertible neural architectures or residual networks with invertibility constraints \citep{chen2024dform,huang2021residual}. For each \(h_\theta\), we approximate the induced pullback composition operator
\[
    \mathcal C_{h_\theta}v = v\circ h_\theta
\]
as a transfer from the \(\Psi\)-dictionary to the \(\Phi\)-dictionary. Concretely, given sample points \(x_1,\ldots,x_N\in X\), we solve
\begin{equation}
\label{eq:CSA_inner}
    P_\theta
    =
    \arg\min_{P\in\mathbb R^{n\times m}}
    \sum_{i=1}^{N}
    \bigl\|
        \Psi(h_\theta(x_i)) - P\Phi(x_i)
    \bigr\|_2^2 .
\end{equation}
Here \(P_\theta\in\mathbb R^{n\times m}\). It need not be square when the two dictionaries have different dimensions. The purpose of this inner problem is to ensure that the alignment matrix is not arbitrary: \(P_\theta\) is the finite-data approximation to the projected pullback induced by \(h_\theta\). In the infinite-data limit, and with complete bases of the Hilbert spaces, Equation \eqref{eq:CSA_inner} recovers the coordinate representation of the corresponding composition operator; see Appendix~\ref{app:CSA_projected_composition_operator} for proofs.

%\paragraph{Outer problem: search over candidate conjugacies.} 
Once we have restricted $P_\theta$ to the class of projected composition operators arising from $\Theta$,  we can optimize over the space of bijections
\begin{equation}
\label{eq:CSA_outer}
    d_{\mathrm{CSA}}(F,G)
    :=
    \inf_{\theta\in\Theta}
    \frac{
        \bigl\|P_\theta F - G P_\theta\bigr\|_F
    }{
        \|P_\theta\|_F
    } .
\end{equation} Thus Eq.~\eqref{eq:CSA_outer} measures the residual of the finite-dimensional intertwining relation
\[
    P_\theta F \approx G P_\theta .
\]
When \(m=n\) and \(P_\theta\) is invertible, this is equivalent to the similarity relation
\[
    G\approx P_\theta F P_\theta^{-1}.
\]
However, CSA does not require \(P_\theta\) to be square or invertible. The residual form is therefore more general than direct matrix similarity. The normalisation by \(\|P_\theta\|_F\) makes the intertwining residual relative to the size of the projected pullback. Pseudocode and optimisation details are given in Appendix~\ref{app:pseudocode}.

\section{Example comparisons between CSA and DSA}
\label{sec:results}

We now test the implications of Sections~\ref{sec:diagnosis} and \ref{sec:CSA} in controlled settings where the relevant state-space transformations and observable bases are known. The controlled setup is essential: by knowing the true bijection $h$ in advance, we can construct pairs of systems that are exactly conjugate by construction and ask whether each method recognizes them as such. As DSA is proposed as an operator-level comparison that should not depend on how the Koopman matrices are learned \citep{ostrow2023beyond}, we instantiate the same test in two pipelines: EDMD with explicit finite dictionaries (\S\ref{sec:results_edmd}), and a Hankel-DMD/Krylov pipeline closer to the original paper's native setting (\S\ref{sec:CSA_krylov}). In each pipeline we run a parameter-indexed family of systems through a known representational change on one side of the comparison --- a state-space conjugacy or a different observable basis --- and ask how faithfully a method preserves the intended pairwise geometry. In every experiment, CSA is given a parametric family that contains the true bijection; this isolates the question of whether the method's construction is well-aligned with the conjugacy criterion, deferring questions of family-misspecification to the discussion.

\subsection{EDMD-based implementation}
\label{sec:results_edmd}

We instantiate the EDMD pipeline on the logistic-map family
\[
    f_k(x)=k\,x(1-x),
    \qquad x\in[0,1],
    \qquad k\in[2,4],
\]
sweeping $k$ across $10$ equally spaced values from $k=2$ (single fixed point) to $k=4$ (chaotic). The logistic map is used only as a controlled scalar-parameter family with non-trivial dynamics; nothing in what follows depends on its specific form. For a chosen observable dictionary $\Phi(x)=(\phi_1(x),\ldots,\phi_m(x))^\top$, we obtain the Koopman matrix $F_k$ representing $\mathcal K_{f_k}$ via EDMD. The dictionary is either Fourier (the constant plus $10$ paired sine--cosine harmonics, $m=25$) or Legendre (the first $m=25$ shifted polynomials), depending on the experiment.

Each logistic parameter $k$ produces one EDMD Koopman matrix in a chosen dictionary and state-space coordinate system. Computing $M\in\{\mathrm{DSA},\mathrm{CSA}\}$ on every pair across the $10$ logistic values yields a $10\times 10$ distance matrix. We build two such matrices: a reference $B$, in which both sides of every pair share the same configuration, and a comparison $D$, in which one side has been moved to a different configuration --- a different dictionary, or a state-space conjugacy applied to the data --- while the other has not. Because every transformed system is, by construction, conjugate to its untransformed counterpart, a conjugacy-invariant method should preserve the reference structure in the comparison: $D\approx B$ throughout. We summarize agreement by the normalized affine RMSE and the Pearson correlation between the entries of $D$ and $B$; precise definitions for each experiment are in Appendix~\ref{app:distance_matrices} and the reference matrices are in Appendix~\ref{app:intra_pairwise_distances}.

\subsubsection{CSA, but not DSA, recognizes systems related by a known state-space conjugacy}
\label{sec:results_edmd_transform}
\begin{figure}
    \centering
    \includegraphics[scale = 0.5]{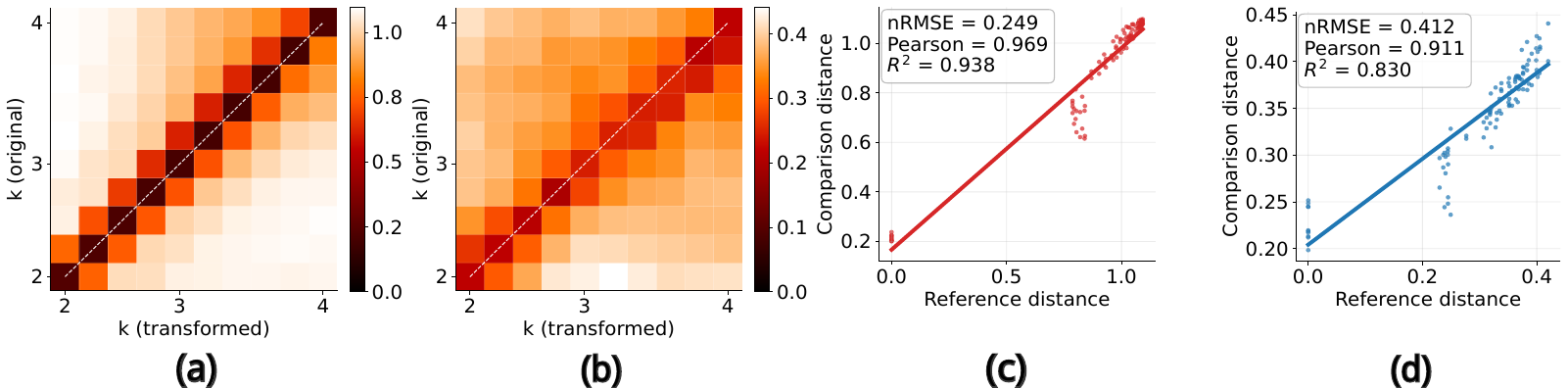}
    \caption{\textbf{State-space transformations with fixed dictionary.}
    Pairwise (comaprison) distance matrices for the two methods - CSA (a) and DSA (b),  when one copy of the logistic-map family is transformed by $h_\theta(x)=x+\theta\sin(\pi x)$ and both Koopman operators are represented in the same Fourier dictionary for both the methods. Scatter plots for reference and comparison matrix entries for CSA (c) and DSA (d).}
    \label{fig:state_transform_heatmap_scatter}
\end{figure}
\begin{figure}
    \centering
    \includegraphics[scale=0.48]{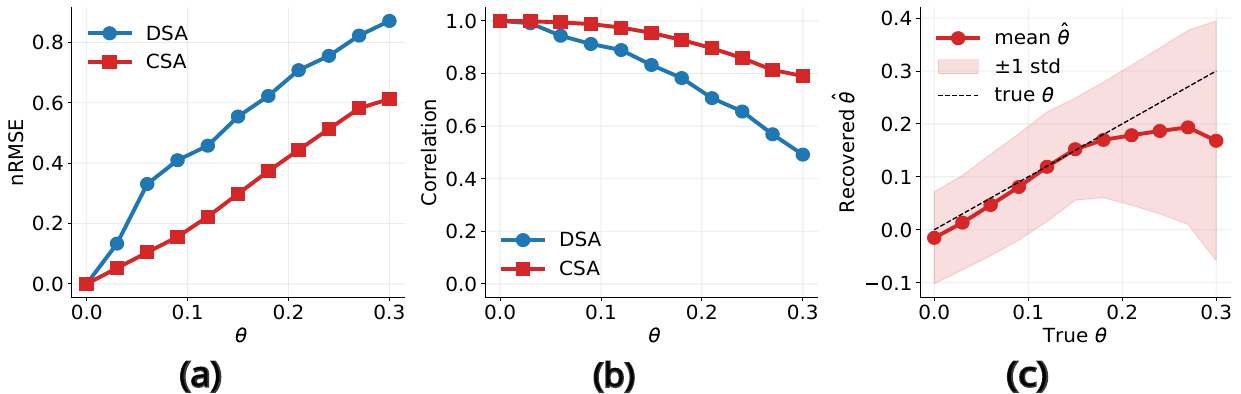}
    \caption{\textbf{Robustness to state-space transformations.} (a)
    nRMSE as a function of $\theta$ for the two methods - CSA (red), DSA (blue). (b) Correlation between entries of the comparison and matrix entries as a function of $\theta$. (c) CSA's recovered conjugacy parameter $\hat \theta$ as a function of $\theta$.}
    \label{fig:state_transform_summary}
\end{figure}
We first test whether each method correctly identifies systems that have been put through a known smooth conjugacy. We hold the dictionary fixed at the Fourier dictionary and vary only the state-space coordinates. From each $f_k$ we form the conjugate twin $\tilde f_{k,\theta} = h_\theta\circ f_k\circ h_\theta^{-1}$ via the smooth bijection $h_\theta(x)=x+\theta\sin(\pi x)$ with $|\theta|<1/\pi$, and fit its EDMD Koopman matrix on the conjugated regression pairs $\{(h_\theta(x_i), h_\theta(f_k(x_i)))\}$ in the same Fourier dictionary. The two sides of the comparison therefore produce Koopman matrices for systems that are exactly conjugate by the known map $h_\theta$, but represented in identical observable coordinates --- so any departure of $D$ from $B$ measures how the comparison method handles the bijection. CSA is given the parametric family $\{h_\theta\}_{\theta\in\mathbb R}$ and searches over $\theta$; DSA searches over orthogonal $Q\in O(m)$.

At a representative $\theta = 0.1$, the comparison matrices $D$ for the two methods, together with their entrywise scatter against $B$, separate the two methods cleanly (Figure~\ref{fig:state_transform_heatmap_scatter}). CSA preserves the diagonal banding of the reference and lies tightly along an affine line in the scatter (Fig.~\ref{fig:state_transform_heatmap_scatter}c), while DSA loses the banding and the scatter is visibly broadened (Fig.~\ref{fig:state_transform_heatmap_scatter}d). Across the full sweep $\theta\in[0,0.30]$, DSA's nRMSE rises rapidly and its Pearson correlation falls below $0.5$ before the sweep ends, whereas CSA-Fourier's nRMSE grows much more slowly and its correlation remains close to $1$ throughout (Fig.~\ref{fig:state_transform_summary}a,b). Beyond reporting a small distance, CSA also identifies the bijection responsible for it: the average $\hat\theta$ of the optimized CSA parameter tracks the diagonal $\hat\theta=\theta$ across the sweep (Fig.~\ref{fig:state_transform_summary}c). Together, these results show that, in this controlled setting, DSA fails as a conjugacy invariant, given that the orthogonal alignment class cannot bridge the non-orthogonal pullback induced by $h_\theta$. CSA, on the other hand, both identifies the conjugacy and quantifies its parameter.

\subsubsection{CSA, but not DSA, recovers equivalence when the same dynamics are written in different bases}
\label{sec:results_edmd_basis}
\begin{figure}
    \centering
    \includegraphics[scale=0.48]{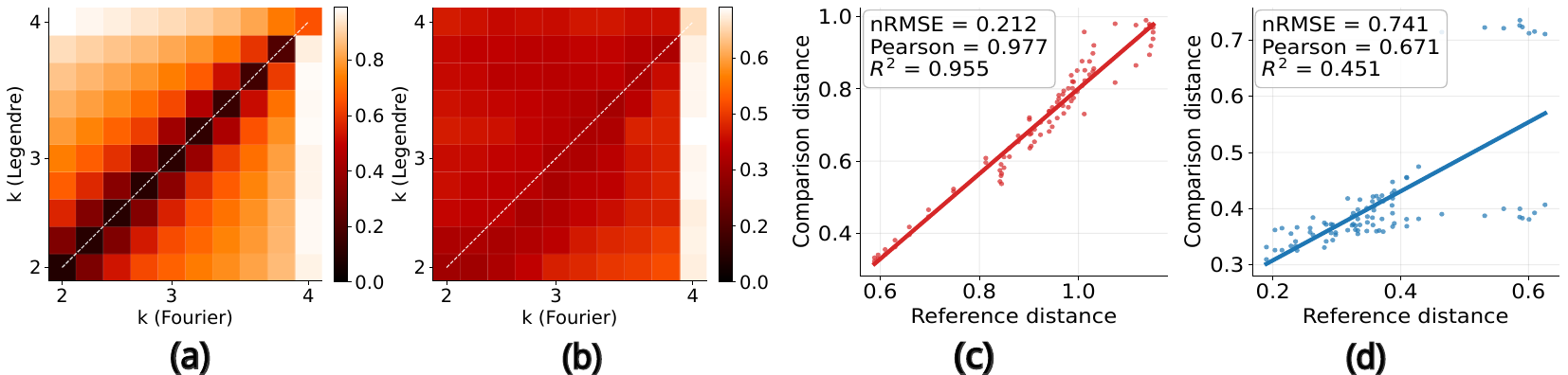}
    \caption{\textbf{Cross-dictionary comparison.}
    Pairwise  (comparison) distance matrices for the two methods - CSA (a) and DSA (b), when the same logistic-map family is
    represented in different dictionaries (Fourier (horizontal axis) vs.\ Legendre (vertical axis)), $\theta = 0$. Scatter plots of reference and comparison matrix entries for CSA (c) and DSA (d).}
    \label{fig:basis_change_heatmap_scatter}
\end{figure}
\begin{figure}
    \centering
    \includegraphics[scale=0.48]{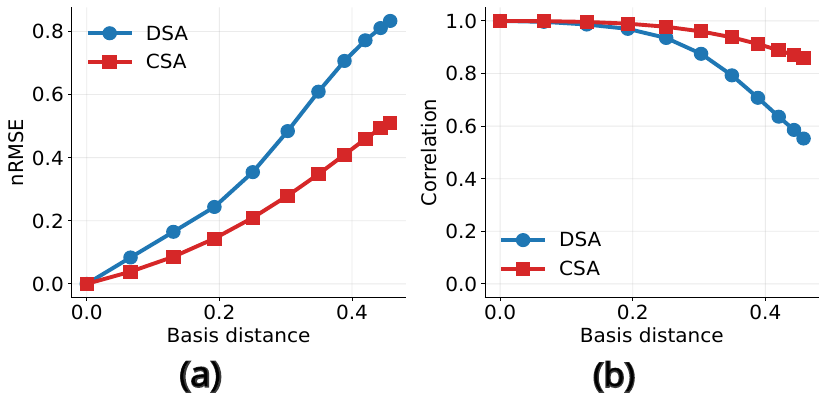}
    \caption{\textbf{Robustness to basis mismatch.}
    (a)
    nRMSE as a function of basis distance (Appendix~\ref{app:basis_distance}) for the two methods - CSA (red), DSA (blue). (b) Correlation between entries of the comparison and matrix entries as a function of distance between the bases. }
    \label{fig:pure_basis_change_summary}
\end{figure}

The first experiment held the dictionary fixed and varied the state-space coordinates. We now do the reverse: we fix the dynamics --- the same logistic system $f_k$ on each side --- and represent its Koopman operator in two distinct dictionaries. The state-space bijection is the identity, $h=\mathrm{id}$, so a conjugacy-invariant method should recognize the two representations as describing the same operator. The cross-basis comparison matrices $D$ between the Legendre and Fourier dictionaries at matched dimension $m=25$ separate the two methods (Fig.~\ref{fig:basis_change_heatmap_scatter}a,b): CSA recovers the reference geometry; DSA does not. The corresponding scatter plots against the reference $B$ (Fig.~\ref{fig:basis_change_heatmap_scatter}c,d) make the gap quantitative: CSA's nRMSE is $0.212$ and its Pearson correlation is $0.977$, while DSA's nRMSE is $0.741$ and its correlation drops to $0.671$. 

We hypothesized that DSA's failure should depend on the degree of misalignment between the two dictionaries' subspaces. To test this, we constructed a parametric family of subspaces interpolating between the Fourier and Legendre endpoints by sampling uniformly along the Grassmannian geodesic between the two subspaces of $L^2([0,1])$ (construction in Appendix~\ref{app:basis_distance}). For each intermediate subspace, we compute the Koopman operators (for all considered values of $k$) in that subspace and the corresponding nRMSE and correlation against the Fourier reference. As the projector distance grows, DSA's nRMSE rises nearly linearly and its Pearson correlation drops sharply, whereas CSA remains close to the reference geometry until the projector distance is large (Fig.~\ref{fig:pure_basis_change_summary}). This pattern is exactly what Section~\ref{sec:diagnosis} predicts: the two dictionaries describe identical operators, but their subspaces in $L^2([0,1])$ are not orthogonally related, so no $Q\in O(m)$ recovers the matched-basis distance. The non-orthogonal $P$ found by CSA does.

\subsection{Hankel-DMD/Krylov-subspace implementation}
\label{sec:CSA_krylov}
\begin{figure}[t]
  \centering
  \includegraphics[scale=0.48]{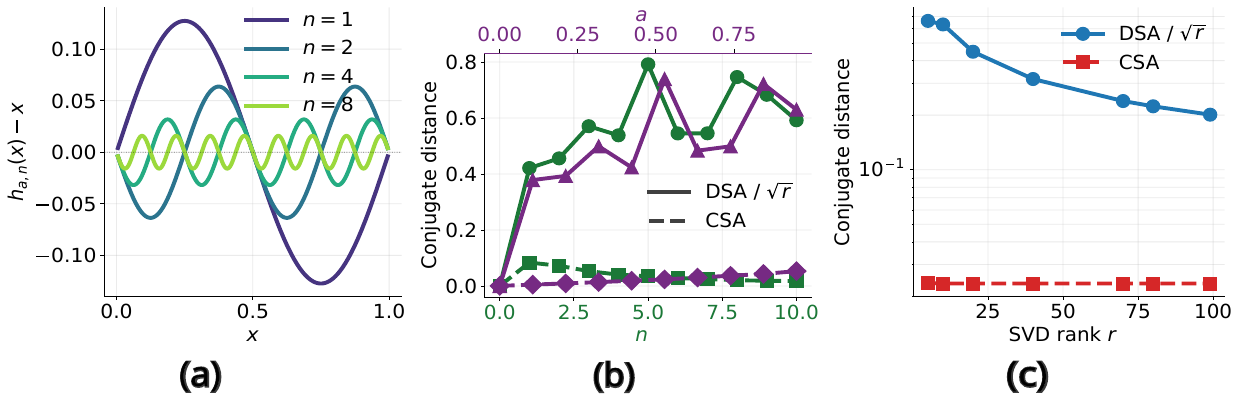}
  \caption{\textbf{Hankel-DMD/Krylov comparison under smooth conjugacies.}
(a) Deviations \(h_{a,n}(x)-x\) for \(a=0.8\). 
(b) Conjugate-pair distances under sweeps over \(n\) at fixed \(a=0.8\) and over \(a\) at fixed \(n=4\), at rank \(r=20\). 
(c) Conjugate-pair distances as a function of rank at \((a,n)=(0.5,4)\).}
  \label{fig:CSA_krylov_section3_2}
\end{figure}
We next move to the setting in which DSA is natively defined: Koopman operators learned from time-delay embeddings. From a scalar trajectory we form a Hankel matrix of delay dimension \(d\), truncate its SVD at rank \(r\), and fit the Koopman matrix in the truncated subspace by least squares (Appendix~\ref{app:krylov}). Crucially, in this pipeline the observable basis is not specified in advance --- it is selected implicitly by the SVD of each system's own Hankel matrix. Two conjugate systems can therefore pick different SVD-truncated bases even when their underlying dynamics are equivalent, putting Hankel-DMD squarely in the regime where Section~\ref{sec:diagnosis} predicts DSA's orthogonal alignment to be too weak. We refer to standard DSA in this setting as \emph{DSA-Krylov} and to its conjugacy-constrained counterpart as \emph{CSA-Krylov}.

Our reference dynamics is the contracting linear map $f_r(y)=\tfrac12+r(y-\tfrac12)$ with $r\in(0,1)$, which we conjugate by the two-parameter $C^\infty$ family of bijections
\[
    h_{a,n}(x)
    =
    x-\frac{a}{2\pi n}\,
    \sin\!\bigl(2\pi n\,(x-\tfrac12)\bigr),
    \qquad
    n\in\mathbb Z_{\ge 1},
    \quad
    a\in[0,1).
\]
Since $h'_{a,n}(x)=1-a\cos\bigl(2\pi n(x-\tfrac12)\bigr)\ge 1-a>0$, each $h_{a,n}$ is a diffeomorphism of $[0,1]$ fixing both endpoints; the amplitude $a$ controls how strongly $h_{a,n}$ departs from the identity, while $n$ sets the number of oscillations of the deviation $h_{a,n}(x)-x$ over the interval (Figure~\ref{fig:CSA_krylov_section3_2}(a)). The conjugate twin of $f_r$ is $g^{(r)}_{a,n}=h_{a,n}^{-1}\circ f_r\circ h_{a,n}$, and we build the four-system grid $\{f_{0.5}, f_{0.8}, g^{(0.5)}_{a,n}, g^{(0.8)}_{a,n}\}$, reporting the average over the two conjugate pairs $(f_r, g^{(r)}_{a,n})$ at $r\in\{0.5, 0.8\}$. For each system we generate $50$ trajectories of length $T=150$ and embed at delay dimension $d=100$; convergence under this truncation is verified in Appendix~\ref{sec:convergence}. CSA-Krylov searches over the same family used to build the twin. For visual comparison on a common scale, we plot $d_{\mathrm{DSA}}/\sqrt{r}$ alongside the CSA-Krylov distance; the $1/\sqrt{r}$ factor is the natural per-dimension normalization of a Frobenius residual and is irrelevant to the conjugacy-invariance question.

We first ask whether each method's distance is flat under varying conjugacies, as a conjugacy invariant should be. Two sweeps at rank $r=20$ --- the number of oscillations $n\in\{0,\ldots,10\}$ at fixed amplitude $a=0.8$, and the amplitude $a\in[0,0.95]$ at fixed $n=4$ --- make the contrast clear (Fig.~\ref{fig:CSA_krylov_section3_2}b). In both sweeps the underlying dynamics is the same contraction $f_r$ --- only the coordinate chart in which it is represented changes --- so any conjugacy invariant should be flat. CSA-Krylov is approximately flat: its distance never exceeds $\sim\!0.08$ over either sweep. DSA-Krylov, by contrast, rises from zero (at $n=0$ or $a=0$, where $h_{a,n}$ is the identity and the two systems coincide) to $0.5$--$0.8$ once the bijection departs appreciably from the identity. DSA-Krylov's distance therefore tracks features of the relabelling which have no bearing on conjugacy.

We next confirm that this gap is structural rather than a truncation artefact. The rank sweep at $(a, n)=(0.5, 4)$ on a log scale (Fig.~\ref{fig:CSA_krylov_section3_2}c) shows that CSA-Krylov is essentially flat at $\sim\!0.02$ across the range $r\in\{5,\ldots,99\}$, while $d_{\mathrm{DSA}}/\sqrt{r}$ decays roughly as $1/\sqrt{r}$ --- that is, the raw DSA distance is itself approximately rank-independent and the apparent decay is supplied entirely by the normalization. Both metrics are therefore well-converged in $r$. The mechanism is the one anticipated by §\ref{sec:diagnosis}: DSA fits each system's SVD independently, so the two truncations pick mismatched observable directions that no orthogonal $Q\in O(r)$ can bridge. CSA-Krylov constructs the Krylov subspace using the truncated SVD basis and solves the appropriate basis-transfer problem.

\section{Discussion}
\label{sec:discussion}

In this paper we showed that Koopman-based similarity analysis depends critically on the structure of the alignment being approximated. If the target notion is conjugacy, then the relevant alignment between observable spaces is not an arbitrary orthogonal transformation, but the pullback induced by a state-space bijection. This distinction matters both theoretically and empirically. Orthogonal alignment can identify similarities that do not correspond to any valid relabelling of states, and it can miss genuine conjugacies when the induced pullback is non-orthogonal or when the observable bases differ. Thus, a small distance between finite Koopman matrices is not, by itself, evidence of shared dynamics unless the admissible alignment class reflects the intended equivalence relation.

CSA instantiates this principle by restricting alignments to projected composition operators induced by candidate bijections. In controlled EDMD and Hankel-DMD settings, this was sufficient to recover conjugacy structure in regimes where DSA was sensitive to basis choice, state-space transformations, or SVD-selected observable subspaces. These experiments were intentionally controlled: CSA was given the relevant basis information and a bijection family containing the true map. The goal was not to present a fully automated comparison pipeline, but to isolate the mathematical issue that any conjugacy-aware pipeline must confront. Holding the learned operators fixed, the difference between CSA and DSA comes from the distance being optimized, and therefore from the structural assumptions each distance makes about equivalence.

This also clarifies what Koopman methods can and cannot buy us. Passing to linear operators is powerful, but it does not eliminate the need to reason about the state-space correspondence. The conjugacy reappears as a composition operator on observables, and its finite-dimensional projection is the object that must align the learned Koopman representations. In this sense, CSA lies between direct state-space matching and unconstrained operator matching: it uses the operator representation, but preserves the state-space map that gives the comparison its dynamical meaning.

Conjugacy is a strong notion of equivalence, and other scientific goals may call for weaker notions. For example, two systems may not be conjugate but may nevertheless share predictive structure under a fixed set of delays, or implement similar input--output maps in a task-relevant subspace. Such similarities may be meaningful, but they should not be confused with conjugacy. Likewise, spectral similarity or orthogonal matrix similarity may reveal useful structure without licensing claims about shared dynamics. The central question is therefore not simply how to design more flexible similarity metrics, but how to match the metric to the equivalence relation that the scientific claim requires.

CSA is not a metric in the formal sense, and its conclusions depend on the candidate bijection family and observable spaces supplied to it. These are limitations of the present instantiation, but they are also part of the point: any method that aims to detect conjugacy must make assumptions about admissible relabellings and about how observables transform across systems. Making those assumptions explicit is preferable to hiding them inside an unconstrained alignment problem.

The main contribution of this work is therefore a clarification with practical consequences. Koopman-based comparisons should not treat finite operator matrices as basis-free objects whose similarity can be interpreted independently of the observable spaces and transformations that produced them; identifying the pullback structure required by conjugacy shows why DSA-style orthogonal alignment can be misleading and points to a principled alternative. This matters wherever conjugacy is the claim being made — across recording sessions for BCI alignment, across animals for cross-subject comparison, between models and data for model validation, or across architectures and training regimes in machine learning. In each case, before asking whether two systems are similar, we must ask in what sense they are meant to be the same.

\bibliography{biblio}

\newpage
\appendix

\section{Unitary operators need not be composition operators}
\label{app:unitary_not_composition}

A composition operator \(\mathcal C_h\) induced by a bijection \(h\) acts on observables by
\[
(\mathcal C_h f)(x)=f(h(x)).
\]
Such operators preserve pointwise products:
\begin{equation}
\mathcal C_h(fg)(x)
=(fg)(h(x))
=f(h(x))\,g(h(x))
=(\mathcal C_h f)(x)\,(\mathcal C_h g)(x).
\label{eq:composition_preserves_products}
\end{equation}
Multiplicativity is a necessary property of composition operators\footnote{ This is precisely the reason why the Eigenfunctions of the Koopman operator form a multiplicative lattice \citep{budivsic2012applied}.}. 

Now consider the two-dimensional subspace
\[
V=\operatorname{span}\{1,x\}\subset L^2([-1,1]),
\]
with orthonormal basis
\[
e_1(x)=\frac{1}{\sqrt{2}},\qquad e_2(x)=\sqrt{\frac{3}{2}}\,x.
\]
Define \(U:V\to V\) by
\[
Ue_1=e_2,\qquad Ue_2=e_1.
\]
Its matrix in the basis \(\{e_1,e_2\}\) is
\[
\begin{pmatrix}
0 & 1\\
1 & 0
\end{pmatrix},
\]
so \(U\) is unitary. However, \(U\) is not a composition operator, because it does not preserve products. Indeed,
\[
e_1e_2=\frac{1}{\sqrt{2}}e_2,
\]
so
\[
U(e_1e_2)=\frac{1}{\sqrt{2}}e_1,
\]
whereas
\[
(Ue_1)(Ue_2)=e_2e_1=e_1e_2=\frac{1}{\sqrt{2}}e_2.
\]
Hence
\[
U(e_1e_2)\neq (Ue_1)(Ue_2).
\]
Therefore \(U\) is unitary but not a composition operator. 

\section{Spectral isomorphism doesn't imply dynamical conjugacy}
\label{app:counterexample}

\begin{proposition}
\label{prop:bernoulli_counterexample_appendix}
Let
\[
X:=\{0,1\}^{\mathbb Z}\times \{0,1\},
\qquad
Y:=\{0,1,2\}^{\mathbb Z}\times \{0,1\},
\]
equipped with the product probability measures
\[
\mu:=\mu_2\times \eta,
\qquad
\nu:=\mu_3\times \eta,
\]
where \(\mu_2\) and \(\mu_3\) are the Bernoulli product measures on
\(\{0,1\}^{\mathbb Z}\) and \(\{0,1,2\}^{\mathbb Z}\), respectively, and \(\eta\) is the uniform measure on \(\{0,1\}\). Define
\[
T:=\sigma_2\times \mathrm{id},
\qquad
S:=\sigma_3\times \mathrm{id},
\]
where \(\sigma_2\) and \(\sigma_3\) are the two-sided Bernoulli shifts. Then the Koopman operators \(\mathcal K_T\) and \(\mathcal K_S\) are unitarily equivalent, but \(T\) and \(S\) are not conjugate.
\end{proposition}

\begin{proof}
We prove the two claims separately.

First, \(\mathcal K_T\) and \(\mathcal K_S\) are unitarily equivalent. For \(k=2,3\), let
\[
X_k:=\{0,1,\dots,k-1\}^{\mathbb Z},
\qquad
(\sigma_kx)_n=x_{n+1},
\]
with \(\mu_k\) the corresponding Bernoulli product measure. It is a standard fact that
\[
L^2(X_k,\mu_k)
=
\mathbb C \mathbf 1
\oplus
\bigoplus_{j=1}^\infty \mathcal H_j^{(k)},
\]
where on each \(\mathcal H_j^{(k)}\), the Koopman operator \(\mathcal K_{\sigma_k}\) acts as a copy of the bilateral shift on \(\ell^2(\mathbb Z)\). Since this decomposition is the same for \(k=2\) and \(k=3\), there exists a unitary
\[
W:L^2(X_2,\mu_2)\to L^2(X_3,\mu_3)
\]
such that
\[
W\mathcal K_{\sigma_2}=\mathcal K_{\sigma_3}W.
\]

Now
\[
L^2(X,\mu)\cong L^2(X_2,\mu_2)\oplus L^2(X_2,\mu_2),
\qquad
L^2(Y,\nu)\cong L^2(X_3,\mu_3)\oplus L^2(X_3,\mu_3),
\]
and under these identifications,
\[
\mathcal K_T=\mathcal K_{\sigma_2}\oplus \mathcal K_{\sigma_2},
\qquad
\mathcal K_S=\mathcal K_{\sigma_3}\oplus \mathcal K_{\sigma_3}.
\]
Hence \(\widetilde W:=W\oplus W\) is unitary and satisfies
\[
\widetilde W \mathcal K_T = \mathcal K_S \widetilde W.
\]
Thus \(\mathcal K_T\) and \(\mathcal K_S\) are unitarily equivalent.

Next, \(T\) and \(S\) are not conjugate. If there were a bijection \(h:X\to Y\) with
\[
h\circ T=S\circ h,
\]
then \(T\) and \(S\) would have the same number of periodic points of each period. But
\[
\#\mathrm{Fix}(\sigma_k^n)=k^n,
\]
since an \(n\)-periodic bi-infinite sequence is determined by a word of length \(n\). Therefore
\[
\#\mathrm{Fix}(T^n)=2\cdot 2^n,
\qquad
\#\mathrm{Fix}(S^n)=2\cdot 3^n.
\]
These are unequal for every \(n\ge 1\), so no such \(h\) can exist. Hence \(T\) and \(S\) are not conjugate.
\end{proof}

\subsection{Why the Bernoulli shifts are spectrally isomorphic}
\label{app:bernoulli_spectral}

For completeness, we briefly indicate the spectral decomposition used above. For \(k=2,3\), write
\[
X_k=\{0,1,\dots,k-1\}^{\mathbb Z},
\qquad
\mathcal H_k=L^2(X_k,\mu_k).
\]
Let
\[
L^2(\{0,1,\dots,k-1\},p_k)
=
\mathbb C\mathbf 1 \oplus \mathcal H_k^0,
\]
where \(p_k\) is the uniform measure on the alphabet, and choose an orthonormal basis
\[
e_1^{(k)},\dots,e_{k-1}^{(k)}
\]
for \(\mathcal H_k^0\).

For \(n\in\mathbb Z\), define the coordinate observables
\[
e_{n,\alpha}^{(k)}(x):=e_\alpha^{(k)}(x_n).
\]
More generally, for a finite set \(J=\{n_1<\cdots<n_r\}\subset\mathbb Z\) and multi-index \(\alpha=(\alpha_1,\dots,\alpha_r)\), define
\[
\psi_{J,\alpha}^{(k)}(x)
:=
\prod_{j=1}^r e_{\alpha_j}^{(k)}(x_{n_j}).
\]
Together with the constant function \(1\), these form an orthonormal basis of \(\mathcal H_k\).

The Koopman operator acts by translation of the index set:
\[
\mathcal K_{\sigma_k}\psi_{J,\alpha}^{(k)}
=
\psi_{J+1,\alpha}^{(k)},
\qquad
J+1:=\{n+1:n\in J\}.
\]
Hence each translation orbit of a finite pattern generates an invariant subspace on which \(\mathcal K_{\sigma_k}\) acts as the bilateral shift on \(\ell^2(\mathbb Z)\). Since there are countably many such orbits for both \(k=2\) and \(k=3\), both operators decompose as
\[
\mathbb C\mathbf 1
\oplus
\bigoplus_{j=1}^\infty \text{(copy of bilateral shift)}.
\]
Thus \(\mathcal K_{\sigma_2}\) and \(\mathcal K_{\sigma_3}\) are unitarily equivalent.

The important point is that this intertwiner is an abstract unitary operator on \(L^2\). It is not induced by a bijection of state spaces. This is precisely why spectral equivalence does not imply conjugacy. It is straightforward to see that CSA doesn't mistakenly qualify them as conjugate as it searches over the space of valid homeomorphisms.

 \begin{remark}There is some existing literature proposing that spectral conjugacy can, in some cases, imply topological conjugacy. \citep{mezic2020spectrum} generates eigenfunctions of a Koopman operator using well-known methods in Morse theory. Denoting as usual our dynamical system by $f:X\xrightarrow{} X$, let $W^+\subset X$ be a stable manifold to a fixed point $x_0$. Let $\phi:W^+ \xrightarrow{} \mathbb{R}^n$ be a (diffeomorphic) coordinate chart, and define $\textbf{y} = \phi(x)$. This chart can be picked such that there is a diagonal matrix $A$ with $\textbf{y}_{t+1} = A\textbf{y}_t$. That is, 
 \begin{equation*}
         \mathcal{K}_f\phi(x_t) = \phi(x_{t+1}) = \textbf{y}_{t+1} = A\textbf{y}_{t}= A\phi(x_{t}).
 \end{equation*}

 Projecting $\phi$ to each eigenspace in $\mathbb{R}^n$ gives us $n$ eigenfunctions. Their span is an invariant subspace of the Hilbert space. However, for this subspace to be orthonormally related to the EDMD/Hankel-DMD subspaces is not always true. Therefore, one cannot naively compute the DSA metric over orthonormal transforms and hope to recover conjugacy in such a scenario.

 \end{remark}

\subsection{Trivial extension to non-ergodic and non-measure-preserving systems}
\label{app:nmp_variant}

The counter-example above can be trivially generalized to non-measure-preserving systems as well. This is done by extending the state-space to include a non-measure-preserving map that is identical for both the systems. For instance,
\[
y_{t+1} = C(y_t):=\frac{y_t}{2}.
\]
Define
\[
\widehat T:=\sigma_2\times C,
\qquad
\widehat S:=\sigma_3\times C.
\]
The lack of conjugacy still follows from the Bernoulli component, since the periodic point counts in the shift coordinate remain different. The point of this variant is simply that the obstruction is not specific to the measure-preserving setting.

\section{When is P unitary?} \label{app:when_is_P_unitary?}
Since \(P\) is the matrix representation of \(\mathcal{C}_h\) in orthonormal bases, \(P\) is unitary if and only if \(\mathcal{C}_h\) is unitary. Equivalently,
\begin{equation}
    P^*P = I
    \qquad \Longleftrightarrow \qquad
    \langle \mathcal{C}_h u,\mathcal{C}_h v\rangle_X = \langle u,v\rangle_Y
    \quad \text{for all } u,v.
\end{equation}

If \(\mathcal H_X=L^2(X,\mu_X)\) and \(\mathcal H_Y=L^2(Y,\mu_Y)\), this is equivalent to
\begin{equation}
    \int_X |u(h(x))|^2\,d\mu_X(x)
    =
    \int_Y |u(y)|^2\,d\mu_Y(y)
    \qquad \text{for all } u,
\end{equation}
or, equivalently,
\begin{equation}
    h_\#\mu_X = \mu_Y.
\end{equation}
Thus \(P\) is unitary precisely when \(h\) preserves the relevant measure. In the special case where \(X,Y\subset\mathbb{R}^d\) and both observable spaces use Lebesgue measure, a smooth invertible map \(h\) is unitary only if
\begin{equation}
    |\det Dh(x)| = 1
    \qquad \text{a.e. on } X.
\end{equation}
Thus many perfectly valid conjugacies do not induce unitary operators. For example, the affine map
\begin{equation}
    h(x)=ax+b
\end{equation}
is unitary on \(L^2\) with Lebesgue measure only when \(|a|=1\). Likewise, the map
\begin{equation}
    h_\theta(x)=x+\theta \sin(\pi x),
    \qquad x\in[0,1],
\end{equation}
has derivative
\begin{equation}
    h_\theta'(x)=1+\theta\pi\cos(\pi x),
\end{equation}
which is not identically of unit modulus unless \(\theta=0\). Therefore, although \(h_\theta\) may be invertible for sufficiently small \(\theta\), the induced operator \(\mathcal{C}_{h_\theta}\) is generally not unitary, and hence neither is its matrix representation \(P_\theta\).

It is therefore important to distinguish invertibility of the conjugacy from unitarity of the induced composition operator. Conjugacy only requires \(h\) to be bijective and to satisfy \(h\circ f = g\circ h\). Unitarity is an additional and much stronger requirement.

As a simple example where \(P\) is unitary, consider a mirror reflection. Let \(X=Y=[-1,1]\), let \(h(x)=-x\), and suppose \(g\) is the reflected version of \(f\), i.e.
\begin{equation}
    g = h\circ f \circ h^{-1}.
\end{equation}
Then
\begin{equation}
    (\mathcal{C}_h u)(x)=u(-x).
\end{equation}
Since the reflection \(x\mapsto -x\) preserves Lebesgue measure,
\begin{align}
    \|\mathcal{C}_h u\|_{L^2([-1,1])}^2
    &= \int_{-1}^1 |u(-x)|^2\,dx
     = \int_{-1}^1 |u(y)|^2\,dy
     = \|u\|_{L^2([-1,1])}^2,
\end{align}
where we used the change of variables \(y=-x\). Hence \(\mathcal{C}_h\) is unitary, and therefore \(P\) is unitary in any pair of orthonormal bases.

If, moreover, the basis on \(X\) is chosen as the reflected image of the basis on \(Y\), i.e.
\begin{equation}
    \phi_i = \mathcal{C}_h \psi_i = \psi_i(-x),
\end{equation}
then the matrix representation of \(\mathcal{C}_h\) is simply
\begin{equation}
    P = I.
\end{equation}
If instead the same basis is used on both sides, then \(P\) need not be the identity, but it will still be unitary. For example, in a polynomial or Fourier basis, reflection acts by parity, so \(P\) becomes diagonal with entries \(\pm 1\).

\section{Construction of comparison and reference distance matrices}
\label{app:distance_matrices}

Let \(\{k_1,\dots,k_N\}\) be the set of logistic-map parameters under consideration. For each \(k_i\), let
\[
K_i^{(\Phi,\theta)}
\]
denote the finite-dimensional Koopman matrix obtained from the system at parameter \(k_i\), represented in dictionary \(\Phi\), with state-space transformation parameter \(\theta\). Here \(\theta=0\) denotes the untransformed system.

For a comparison method \(M\in\{\mathrm{DSA},\mathrm{CSA}\}\), the associated pairwise distance matrix is
\[
D^{M}\bigl[(\Phi_1,\theta_1),(\Phi_2,\theta_2)\bigr]\in\mathbb R^{N\times N},
\]
with entries
\[
D^{M}\bigl[(\Phi_1,\theta_1),(\Phi_2,\theta_2)\bigr]_{ij}
=
d_M\!\left(
K_i^{(\Phi_1,\theta_1)},
K_j^{(\Phi_2,\theta_2)}
\right).
\]

The matrix denoted by \(D\) in the main text is the comparison matrix relevant to the experiment under consideration, while \(B\) is the corresponding matched reference matrix.

\paragraph{State-space transformation experiment.}
When the dictionary is fixed to \(\Phi\) and one side is transformed by \(h_\theta\), the comparison matrix is
\[
D
=
D^{M}\bigl[(\Phi,0),(\Phi,\theta)\bigr],
\]
and the matched reference matrix is
\[
B
=
D^{M}\bigl[(\Phi,0),(\Phi,0)\bigr].
\]

\paragraph{Cross-dictionary experiment.}
When the underlying dynamics are unchanged but the dictionaries differ, the comparison matrix is
\[
D
=
D^{M}\bigl[(\Phi,0),(\Psi,0)\bigr].
\]
A natural matched reference is the symmetrized within-dictionary baseline
\[
B
=
\frac12\left(
D^{M}\bigl[(\Phi,0),(\Phi,0)\bigr]
+
D^{M}\bigl[(\Psi,0),(\Psi,0)\bigr]
\right).
\]
This symmetrization removes an arbitrary preference for one of the two endpoint dictionaries.

\paragraph{Matrix-level summary statistics.}
Given a comparison matrix \(D\) and reference matrix \(B\), we flatten their upper-triangular entries and fit the best affine map from \(B\) to \(D\). The reported normalized affine RMSE measures the residual error after this affine alignment, and the Pearson correlation quantifies the extent to which the pairwise geometry is preserved up to affine rescaling.

\section{Distance between observable subspaces}
\label{app:basis_distance}

This appendix details two ingredients of the basis-change experiment of
Section~\ref{sec:results_edmd_basis}: the construction of a one-parameter family of
$m$-dimensional observable subspaces interpolating between the Fourier and
Legendre endpoint dictionaries, and the projector distance used to label
each member of that family on the horizontal axis of
Figure~\ref{fig:pure_basis_change_summary}.

\paragraph{Subspace family between Fourier and Legendre.}
The two endpoint dictionaries are taken at matched dimension $m=25$: the
Fourier dictionary contains the constant together with the first ten
cosine/sine harmonics, and the Legendre dictionary contains the first $m$
Legendre polynomials on $[0,1]$.  We embed both into a common ambient
dictionary by simple concatenation, of dimension $p=2m$, and evaluate it
on a dense uniform grid $\{x_\ell\}_{\ell=1}^N$ of $N=4000$ points in
$[0,1]$ to obtain the ambient feature matrix
$A\in\mathbb R^{N\times p}$. The empirical Gram matrix
\begin{equation}\label{eq:emp_gram}
\hat G \;=\; \frac{1}{N}\,A^\top A \;\in\;\mathbb R^{p\times p}
\end{equation}
is the discrete approximation of the $L^2([0,1])$ inner product
$\langle \phi_i,\phi_j\rangle = \int_0^1 \phi_i(x)\phi_j(x)\,\mathrm dx$
on the ambient features under Lebesgue measure on $[0,1]$.  All inner
products and orthogonality statements below are with respect to $\hat
G$; the dense uniform grid sets the discretisation, and refining it
recovers the continuum $L^2$ pairing.

Let $S=\hat G^{1/2}$.  Pre-multiplication by $S$ is the change of variables
that turns $\hat G$-orthonormality into Euclidean orthonormality: a
coefficient matrix $C\in\mathbb R^{p\times m}$ satisfies $C^\top \hat G
C=I_m$ iff $U=SC$ satisfies $U^\top U=I_m$.  Let $C_F,C_L\in\mathbb
R^{p\times m}$ be the selector matrices that pick the Fourier and Legendre
blocks of the ambient dictionary, re-orthonormalised in the $\hat G$ inner
product.  The two endpoint subspaces are then represented by the
Euclidean-orthonormal frames
\begin{equation}\label{eq:endpoint_frames}
U_F = SC_F,\qquad U_L = SC_L \;\in\;\mathbb R^{p\times m}.
\end{equation}

For $t\in[0,1]$ the intermediate frame is the principal geodesic from
$U_F$ to $U_L$ on the Grassmannian $\mathrm{Gr}(m,p)$,
\begin{equation}\label{eq:geodesic}
U(t) \;=\; U_F P \cos(t\Theta) + W\sin(t\Theta),
\end{equation}
where $U_F^\top U_L = P\cos\Theta\,Q^\top$ is the thin SVD,
$\Theta=\mathrm{diag}(\theta_1,\ldots,\theta_m)$ collects the principal
angles between $U_F$ and $U_L$, and
$W=(U_L Q-U_F P\cos\Theta)\sin^{-1}\Theta$ is the orthonormal completion.
By construction $U(0)=U_F$ and $U(1)=U_L$, and $U(t)$ remains
Euclidean-orthonormal in the ambient space throughout.  The corresponding
intermediate dictionary, in $L^2$-coefficient form, is recovered by
inverting the change of variables, $C(t) = S^{-1}U(t)$, and re-orthonormalising
in the $\hat G$ inner product; on any evaluation grid the dictionary
itself is $\Phi^{(t)}(x) = A_{\text{ambient}}(x)\,C(t)$.  The Koopman
operators reported in Figure~\ref{fig:pure_basis_change_summary} are
fitted by the same EDMD procedure of Section~\ref{sec:results_edmd} using
$\Phi^{(t)}$ as the dictionary.

\paragraph{Projector distance between subspaces.}
Given two Euclidean-orthonormal frames $U_1,U_2\in\mathbb R^{p\times m}$
in the ambient feature space, the basis distance we report is the
normalised Frobenius distance between the corresponding orthogonal
projectors,
\begin{equation}\label{eq:basis_distance}
d_{\mathrm{basis}}(U_1,U_2)
\;=\;
\frac{\|U_1U_1^\top - U_2U_2^\top\|_F}{\sqrt{2m}}.
\end{equation}
It vanishes when the two $m$-dimensional subspaces coincide and equals
$1$ when they are mutually orthogonal.  In terms of the principal angles
$\theta_1,\ldots,\theta_m$ between the subspaces it admits the closed
form
\(
d_{\mathrm{basis}}^2 = \tfrac{1}{m}\sum_{i=1}^m\sin^2\theta_i,
\)
so it coincides, up to the $1/\sqrt{m}$ normalisation, with the chordal
Grassmann distance on $\mathrm{Gr}(m,p)$.  Because the endpoint frames
\eqref{eq:endpoint_frames} carry the empirical Gram square root $S$, the
Euclidean projector distance \eqref{eq:basis_distance} computed on
$U_F,U(t),U_L$ is exactly the projector distance between the underlying
function subspaces in $L^2([0,1])$ under the empirical inner product
\eqref{eq:emp_gram}.  The horizontal axis of
Figure~\ref{fig:pure_basis_change_summary} is
$d_{\mathrm{basis}}(U_F,U(t))$, which increases monotonically with $t$
from zero at $t=0$ to the endpoint distance $d_{\mathrm{basis}}(U_F,U_L)$
at $t=1$.

\section{Intra-transformation distance matrices} \label{app:intra_pairwise_distances}
In Figure \ref{fig:intra_pairwise_matrices} we show that both DSA and CSA are able to identify conjugacy across two Logistic systems when both of those systems are represented in the same basis $\Phi$ (either Fourier or Legendre) and they share the same state-space (either $x$ (called original) or $\hat x = h_\theta (x)$ (called transformed)). These matrices presented in this figure are used to generate the reference matrices in Appendix \ref{app:distance_matrices}. In particular we have plotted \[
D^{M}\bigl[(\Phi,\theta),(\Phi,\theta)\bigr]\in\mathbb R^{N\times N},
\] for $M \in \{\text{DSA},\text{CSA}\}$.
\begin{figure}
    \centering
    \includegraphics[scale=0.29]{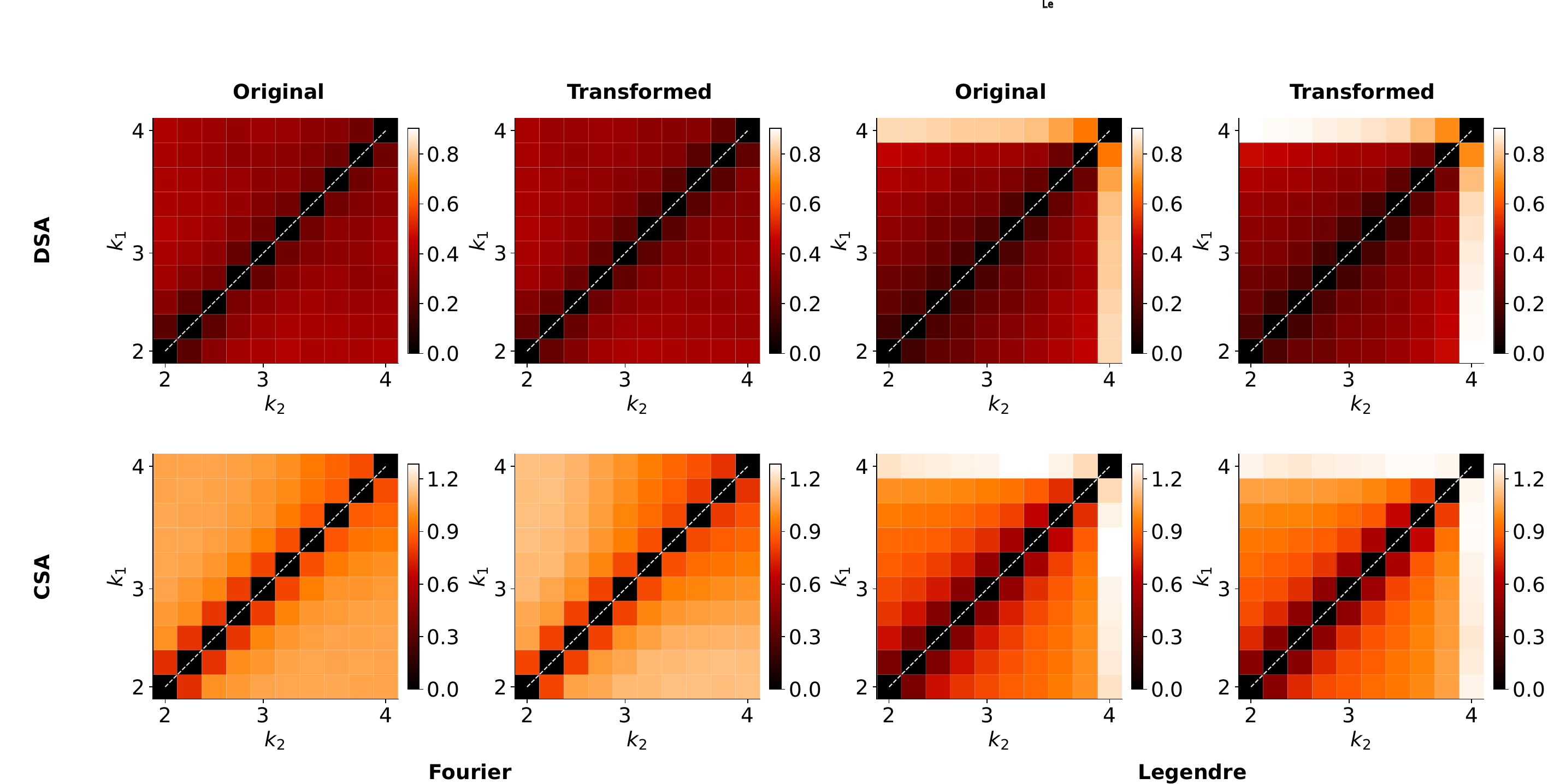}
    \caption{Pairwise distance matrices of within transformation pairs for both DSA (top) and CSA (bottom). $k_i$ represents the $k$ parameter of the Logistic map for the $i$th element in the pair as in the main text.}
    \label{fig:intra_pairwise_matrices}
\end{figure}

\section{Krylov subspace construction via delay-embedding SVD}
\label{app:krylov}

This appendix details how the Koopman operators compared by DSA-Krylov and
CSA-Krylov in Section~\ref{sec:CSA_krylov} are constructed from scalar
time-series data.  The
procedure follows the Hankel-DMD
framework~\citep{arbabi2017ergodic,ostrow2023beyond}, which we summarise here
for completeness and to fix notation.

% ──────────────────────────────────────────────────────────
\subsection{Delay-embedding and the Hankel matrix}

Let $f\colon [0,1]\to[0,1]$ be a one-dimensional map with Koopman operator
$\mathcal{U}$, and let $g(x) = x$ be the identity observable.  From a scalar
time series $\{x_0, x_1, \dots, x_{T-1}\}$ with $x_{t+1} = f(x_t)$, we form
the \emph{delay-embedding} (Hankel) matrix with embedding dimension~$d$:
\begin{equation}\label{eq:hankel}
  H = \begin{pmatrix}
    x_0     & x_1     & \cdots & x_{d-1}   \\
    x_1     & x_2     & \cdots & x_d       \\
    \vdots  & \vdots  & \ddots & \vdots    \\
    x_{T-d} & x_{T-d+1} & \cdots & x_{T-1}
  \end{pmatrix}
  \;\in\; \mathbb{R}^{N \times d},
  \quad N = T - d + 1.
\end{equation}
Each row is a delay-embedded state vector
$z_t = (x_t, x_{t+1}, \dots, x_{t+d-1})^\top \in \mathbb{R}^d$,
and each column~$j$ contains the Krylov observable
$\mathcal{U}^j g$ evaluated along the orbit:
$H_{\cdot,j} = (g(x_j), g(x_{j+1}), \dots, g(x_{j+N-1}))^\top$.
The column space of~$H$ therefore spans the data-driven Krylov subspace
$\mathcal{K}_d(\mathcal{U}, g) =
  \operatorname{span}\{g, \mathcal{U}g, \dots, \mathcal{U}^{d-1}g\}$.

When multiple independent trajectories are available, each is embedded
separately (so that no delay window crosses a trajectory boundary) and the
resulting matrices are stacked vertically, increasing~$N$ while
keeping~$d$ fixed.

% ──────────────────────────────────────────────────────────
\subsection{SVD truncation}

We compute the singular value decomposition
\begin{equation}\label{eq:svd}
  H = U\,\Sigma\,V^\top,
\end{equation}
with $U \in \mathbb{R}^{N \times d}$,
$\Sigma = \operatorname{diag}(\sigma_1 \geq \sigma_2 \geq \cdots \geq \sigma_d)$,
and $V \in \mathbb{R}^{d \times d}$ orthogonal.  Truncating at rank~$r \leq d$
defines the $r$-dimensional subspace
\begin{equation}\label{eq:Vr}
  V_r = V_{:,\,1{:}r} \;\in\; \mathbb{R}^{d \times r},
\end{equation}
whose columns form an orthonormal basis for the leading~$r$
right-singular directions.  The fraction of total energy captured is
$\sum_{i=1}^r \sigma_i^2 \big/ \sum_{i=1}^d \sigma_i^2$.

The projected Krylov observables at a delay-embedded point~$z$ are
\begin{equation}\label{eq:proj}
  \phi(z) = V_r^\top z \;\in\; \mathbb{R}^r.
\end{equation}

% ──────────────────────────────────────────────────────────
\subsection{Koopman operator in the truncated subspace}

Partition the delay vectors into current and next-step pairs:
\begin{equation}
  Z = (z_0, z_1, \dots, z_{N-2})^\top \in \mathbb{R}^{(N{-}1)\times d},
  \quad
  Z' = (z_1, z_2, \dots, z_{N-1})^\top \in \mathbb{R}^{(N{-}1)\times d},
\end{equation}
and project both onto the truncated subspace, $A=ZV_r$ and $B=Z'V_r$ in
$\mathbb{R}^{(N{-}1)\times r}$, so that the rows of $A$ and $B$ are
$\phi(z_t)^\top$ and $\phi(z_{t+1})^\top$ respectively.  Writing the
Koopman matrix $K\in\mathbb R^{r\times r}$ in the column-vector
convention $\phi(z_{t+1})\approx K\phi(z_t)$, the regularised
least-squares fit minimises $\|AK^\top-B\|_F^2+\lambda\|K^\top\|_F^2$,
\begin{equation}\label{eq:koopman_fit}
  K^\top = \bigl(A^\top A + \lambda I_r\bigr)^{-1} A^\top B.
\end{equation}

% ──────────────────────────────────────────────────────────
\subsection{DSA-Krylov: orthogonal alignment}

With $K_1, K_2 \in \mathbb{R}^{r\times r}$ fitted via~\eqref{eq:koopman_fit},
DSA-Krylov seeks the orthogonal matrix that best aligns them:
\begin{equation}\label{eq:dsa}
  d_{\mathrm{DSA\text{-}Krylov}}(K_1, K_2)
    = \min_{Q \in O(r)} \lVert K_1 - Q\,K_2\,Q^\top \rVert_F,
\end{equation}
solved by Riemannian optimisation on $O(r)$ with spectral
initialisation~\citep{ostrow2023beyond}.

Since the two systems may have been fitted in different SVD bases $V_{1,r}$
and $V_{2,r}$, the true basis change need not be orthogonal even when the
systems are dynamically equivalent. Restricting to $O(r)$ therefore conflates
coordinate artefacts with dynamical differences.

% ──────────────────────────────────────────────────────────
\subsection{CSA-Krylov: conjugacy-constrained basis change}

CSA-Krylov replaces the orthogonal search with a basis change \emph{derived}
from a candidate conjugacy $h_\theta$, in two steps.

\paragraph{Step 1: regression for $P_\theta$.}
Fix a dense grid $\{x_i\}_{i=1}^M \subset (0,1)$ and evaluate the projected
Krylov observables for both systems:
\begin{align}
  \Phi_1 &= \bigl[\phi_1(x_1), \dots, \phi_1(x_M)\bigr]^\top
           \in \mathbb{R}^{M \times r},
  \label{eq:phi1} \\
  \Phi_2 &= \bigl[\phi_2(h_\theta(x_1)), \dots,
             \phi_2(h_\theta(x_M))\bigr]^\top
           \in \mathbb{R}^{M \times r},
  \label{eq:phi2}
\end{align}
where $\phi_i(x) = V_{i,r}^\top\,(x, f_i(x), \dots, f_i^{d-1}(x))$ uses
the known dynamics to build Krylov observables at arbitrary evaluation points.
The basis-change matrix is then
\begin{equation}\label{eq:P_regression}
  P_\theta = \bigl[(\Phi_1^\top \Phi_1 + \lambda I)^{-1}\,
             \Phi_1^\top\,\Phi_2\bigr]^\top
           \in \mathbb{R}^{r \times r}.
\end{equation}

\paragraph{Step 2: intertwining residual.}
The CSA-Krylov distance is the normalised intertwining error minimised over
the conjugacy family:
\begin{equation}\label{eq:CSA_krylov}
  d_{\mathrm{CSA\text{-}Krylov}}(K_1, K_2)
    = \min_\theta \;
      \frac{\lVert P_\theta\,K_1 - K_2\,P_\theta \rVert_F}
           {\lVert P_\theta \rVert_F}.
\end{equation}
The minimum is found by coarse grid search followed by L-BFGS-B refinement.

\section{The CSA transfer matrix as a projected composition operator}
\label{app:CSA_projected_composition_operator}

We now make precise the sense in which the CSA transfer matrix is a finite-dimensional
projection of the composition operator induced by a candidate state-space map. This is the
direct analogue, for the basis-transfer operator, of the standard EDMD fact that the fitted
Koopman matrix is an empirical projection of the true Koopman operator and converges to
the corresponding population projection as the number of samples grows (see \citep{korda2018convergence}).

Let $(X,\mu_X)$ and $(Y,\mu_Y)$ be probability spaces, and let
\[
\mathcal H_X=L^2(X,\mu_X),\qquad
\mathcal H_Y=L^2(Y,\mu_Y).
\]
Let $h_\theta:X\to Y$ be measurable, and define the pullback, or composition, operator
\[
\mathcal C_{h_\theta}:\mathcal H_Y\to \mathcal H_X,
\qquad
(\mathcal C_{h_\theta}u)(x)=u(h_\theta(x)).
\]
Throughout this subsection, we assume that $\mathcal C_{h_\theta}\psi_j\in L^2(X,\mu_X)$ for the
finite collection of observables used below. For the complete-basis convergence statement
at the end, we assume the stronger condition that $\mathcal C_{h_\theta}$ is a bounded operator from
$\mathcal H_Y$ to $\mathcal H_X$.

Let
\[
V_\Phi=\operatorname{span}\{\phi_1,\ldots,\phi_m\}\subset \mathcal H_X,
\qquad
V_\Psi=\operatorname{span}\{\psi_1,\ldots,\psi_n\}\subset \mathcal H_Y,
\]
and write
\[
\Phi(x)=
\begin{bmatrix}
\phi_1(x)\\
\vdots\\
\phi_m(x)
\end{bmatrix},
\qquad
\Psi(y)=
\begin{bmatrix}
\psi_1(y)\\
\vdots\\
\psi_n(y)
\end{bmatrix}.
\]
We assume for simplicity that the observables are real-valued. The complex-valued case is
identical after replacing transposes by Hermitian transposes.

Given sample points $x_1,\ldots,x_M\in X$, define the empirical measure
\[
\widehat\mu_M=\frac1M\sum_{\ell=1}^M\delta_{x_\ell}.
\]
The CSA basis-transfer matrix for a fixed candidate map $h_\theta$ is defined by
\begin{equation}
P_{\theta,M}
=
\arg\min_{P\in\mathbb R^{n\times m}}
\frac1M\sum_{\ell=1}^M
\left\|
\Psi(h_\theta(x_\ell))-P\Phi(x_\ell)
\right\|_2^2 .
\label{eq:CSA_empirical_transfer}
\end{equation}
Here we use the row convention: the $j$th row of $P_{\theta,M}$ contains the coefficients,
in the $\Phi$-dictionary, of the approximation to the pulled-back observable
$\mathcal C_{h_\theta}\psi_j=\psi_j\circ h_\theta$. Thus $P_{\theta,M}$ is a row-coefficient matrix.
If one instead represents linear maps by column coefficients, the corresponding operator
matrix is $P_{\theta,M}^\top$.

\begin{proposition}[Finite-data transfer as an empirical projection]
\label{prop:finite_data_transfer_projection}
Assume that the empirical Gram matrix
\[
G_{\Phi,M}
=
\frac1M\sum_{\ell=1}^M
\Phi(x_\ell)\Phi(x_\ell)^\top
\]
is invertible. Let $\Pi_\Phi^{\widehat\mu_M}$ denote the
$L^2(\widehat\mu_M)$-orthogonal projection onto $V_\Phi$. Then, for each
$j=1,\ldots,n$,
\[
\sum_{i=1}^m (P_{\theta,M})_{ji}\phi_i
=
\Pi_\Phi^{\widehat\mu_M} \mathcal C_{h_\theta}\psi_j .
\]
Equivalently, $P_{\theta,M}$ is the row-coefficient matrix of the empirical projected
composition operator
\[
\Pi_\Phi^{\widehat\mu_M} \mathcal C_{h_\theta}\big|_{V_\Psi}:V_\Psi\to V_\Phi .
\]
\end{proposition}

\begin{proof}
The least-squares problem in \eqref{eq:CSA_empirical_transfer} separates over the rows of
$P_{\theta,M}$. Let $p_{j,M}^\top$ denote the $j$th row. Then $p_{j,M}$ solves
\[
\min_{p\in\mathbb R^m}
\frac1M\sum_{\ell=1}^M
\left(
\psi_j(h_\theta(x_\ell))-p^\top\Phi(x_\ell)
\right)^2 .
\]
Since $G_{\Phi,M}$ is invertible, the minimizer is unique and satisfies
\[
p_{j,M}
=
G_{\Phi,M}^{-1}
\frac1M\sum_{\ell=1}^M
\Phi(x_\ell)\psi_j^\top(h_\theta(x_\ell)).
\]
On the other hand, the empirical projection of $\mathcal C_{h_\theta}\psi_j$ onto $V_\Phi$ is,
by definition,
\[
\Pi_\Phi^{\widehat\mu_M}\mathcal C_{h_\theta}\psi_j
=
\arg\min_{f\in V_\Phi}
\|f-\mathcal C_{h_\theta}\psi_j\|_{L^2(\widehat\mu_M)} .
\]
Writing $f=p^\top\Phi$, this becomes exactly
\[
\arg\min_{p\in\mathbb R^m}
\frac1M\sum_{\ell=1}^M
\left(
p^\top\Phi(x_\ell)-\psi_j(h_\theta(x_\ell))
\right)^2 .
\]
Thus
\[
p_{j,M}^\top\Phi
=
\Pi_\Phi^{\widehat\mu_M}\mathcal C_{h_\theta}\psi_j .
\]
Applying this argument for each $j=1,\ldots,n$ proves the claim.
\end{proof}

We next pass from the empirical projection to the population projection.

\begin{proposition}[Infinite-data limit for fixed dictionaries]
\label{prop:infinite_data_transfer_projection}
Assume that $x_1,\ldots,x_M$ are i.i.d. samples from $\mu_X$ and that
\[
G_\Phi
=
\int_X \Phi(x)\Phi(x)^\top\,d\mu_X(x)
\]
is invertible. Assume also that the entries of $\Phi(x)\Phi(x)^\top$ and
$\Phi(x)\psi_j(h_\theta(x))$ are integrable for each $j=1,\ldots,n$.

Define the population transfer matrix
\begin{equation}
P_\theta
=
\arg\min_{P\in\mathbb R^{n\times m}}
\int_X
\left\|
\Psi(h_\theta(x))-P\Phi(x)
\right\|_2^2\,d\mu_X(x).
\label{eq:CSA_population_transfer}
\end{equation}
Then
\[
P_{\theta,M}\to P_\theta
\]
almost surely in any finite-dimensional matrix norm. Moreover, for each $j=1,\ldots,n$,
\[
\sum_{i=1}^m (P_\theta)_{ji}\phi_i
=
\Pi_\Phi^{\mu_X} \mathcal C_{h_\theta}\psi_j ,
\]
where $\Pi_\Phi^{\mu_X}$ is the $L^2(\mu_X)$-orthogonal projection onto $V_\Phi$.
Equivalently, $P_\theta$ is the row-coefficient matrix of the population projected
composition operator
\[
\Pi_\Phi^{\mu_X} \mathcal C_{h_\theta}\big|_{V_\Psi}:V_\Psi\to V_\Phi .
\]
\end{proposition}

\begin{proof}
The population least-squares problem also separates row by row. The $j$th row of
$P_\theta$, denoted $p_j^\top$, is the unique minimizer of
\[
\min_{p\in\mathbb R^m}
\int_X
\left(
\psi_j(h_\theta(x))-p^\top\Phi(x)
\right)^2
\,d\mu_X(x).
\]
Its normal equations give
\[
p_j
=
G_\Phi^{-1}
\int_X
\Phi(x)\psi_j^\top(h_\theta(x))\,d\mu_X(x).
\]
Similarly, the $j$th row of $P_{\theta,M}$ satisfies
\[
p_{j,M}
=
G_{\Phi,M}^{-1}
\frac1M\sum_{\ell=1}^M
\Phi(x_\ell)\psi_j^\top(h_\theta(x_\ell)).
\]
By the strong law of large numbers,
\[
G_{\Phi,M}\to G_\Phi
\]
and
\[
\frac1M\sum_{\ell=1}^M
\Phi(x_\ell)\psi_j(h_\theta(x_\ell))
\to
\int_X
\Phi(x)\psi_j^\top(h_\theta(x))\,d\mu_X(x)
\]
almost surely. Since matrix inversion is continuous on the set of invertible matrices,
$p_{j,M}\to p_j$ almost surely for each $j$. There are finitely many rows, so
$P_{\theta,M}\to P_\theta$ almost surely in any matrix norm.

Finally, the identity
\[
\sum_{i=1}^m (P_\theta)_{ji}\phi_i
=
\Pi_\Phi^{\mu_X} \mathcal C_{h_\theta}\psi_j
\]
follows because the population projection
\[
\Pi_\Phi^{\mu_X}\mathcal C_{h_\theta}\psi_j
=
\arg\min_{f\in V_\Phi}
\|f-\mathcal C_{h_\theta}\psi_j\|_{L^2(\mu_X)}
\]
is obtained by the same normal equations.
\end{proof}

Thus, for fixed dictionaries, CSA estimates the finite-dimensional operator
\[
\Pi_\Phi^{\mu_X} \mathcal C_{h_\theta}\big|_{V_\Psi}:V_\Psi\to V_\Phi.
\]
This is the precise sense in which the fitted transfer matrix is a projected composition
operator. It is not an arbitrary linear map between coefficient spaces, and it is not merely
an orthogonal change of basis. It is the matrix representation of the pullback by $h_\theta$,
restricted to the source dictionary and projected into the target dictionary.

We finally record the complete-basis limit. This statement explains how the projected
composition operators recover the true composition operator as the dictionaries become
complete.

\begin{proposition}[Complete-basis limit]
\label{prop:complete_basis_composition_convergence}
Suppose that $\{\phi_i\}_{i\geq1}$ and $\{\psi_j\}_{j\geq1}$ are orthonormal bases of
$\mathcal H_X$ and $\mathcal H_Y$, respectively. Let
\[
V_m^X=\operatorname{span}\{\phi_1,\ldots,\phi_m\},
\qquad
V_n^Y=\operatorname{span}\{\psi_1,\ldots,\psi_n\},
\]
and let $\Pi_m^X:\mathcal H_X\to V_m^X$ and
$\Pi_n^Y:\mathcal H_Y\to V_n^Y$ denote the corresponding orthogonal projections.
If
\[
\mathcal C_{h_\theta}:\mathcal H_Y\to \mathcal H_X
\]
is bounded, then
\[
\Pi_m^X \mathcal C_{h_\theta}\Pi_n^Y
\to
\mathcal C_{h_\theta}
\]
in the strong operator topology as $m,n\to\infty$. That is, for every
$u\in\mathcal H_Y$,
\[
\left\|
\Pi_m^X \mathcal C_{h_\theta}\Pi_n^Y u
-
\mathcal C_{h_\theta}u
\right\|_{\mathcal H_X}
\to 0.
\]
\end{proposition}

\begin{proof}
Fix $u\in\mathcal H_Y$. Then
\[
\begin{aligned}
\left\|
\Pi_m^X \mathcal C_{h_\theta}\Pi_n^Y u
-
\mathcal C_{h_\theta}u
\right\|_{\mathcal H_X}
&\leq
\left\|
\Pi_m^X \mathcal C_{h_\theta}(\Pi_n^Y u-u)
\right\|_{\mathcal H_X}
+
\left\|
(\Pi_m^X-I)\mathcal C_{h_\theta}u
\right\|_{\mathcal H_X} \\
&\leq
\|\mathcal C_{h_\theta}\|\,
\|\Pi_n^Y u-u\|_{\mathcal H_Y}
+
\left\|
(\Pi_m^X-I)\mathcal C_{h_\theta}u
\right\|_{\mathcal H_X}.
\end{aligned}
\]
Because $\{\psi_j\}_{j\geq1}$ is an orthonormal basis of $\mathcal H_Y$,
$\Pi_n^Y u\to u$ in $\mathcal H_Y$. Because $\{\phi_i\}_{i\geq1}$ is an orthonormal
basis of $\mathcal H_X$, $\Pi_m^X v\to v$ in $\mathcal H_X$ for every
$v\in\mathcal H_X$, and in particular for $v=\mathcal C_{h_\theta}u$. Hence both terms on
the right-hand side converge to zero.
\end{proof}

Combining Propositions~\ref{prop:finite_data_transfer_projection},
\ref{prop:infinite_data_transfer_projection}, and
\ref{prop:complete_basis_composition_convergence}, we obtain the following interpretation:
for fixed dictionaries, the CSA least-squares transfer is the empirical projection of the
composition operator induced by $h_\theta$; as the number of samples grows, it converges
to the population projected composition operator; and, under complete orthonormal
dictionaries and boundedness of the pullback operator, these projected operators converge
strongly to the true composition operator.

\section{Pseudocode and optimisation details}
\label{app:pseudocode}

We describe the computational procedures for CSA (Section~\ref{sec:CSA}) and
CSA-Krylov (Section~\ref{sec:CSA_krylov}). Both share a two-level structure:
an \emph{inner} regression that fits the basis-transfer matrix~$P_\theta$ for
each candidate~$\theta$, and an \emph{outer} search over $\theta$ that
minimises the intertwining residual.

% ─────────────────────────────────────────────────────────
\subsection{CSA}

In the EDMD setting the observable dictionaries $\Phi$ and $\Psi$ are known
explicitly, so the regression for $P_\theta$ reduces to a standard
least-squares problem.

\begin{algorithm}[H]
\caption{Conjugacy-based Similarity Analysis (CSA)}
\label{alg:CSA}
\begin{algorithmic}[1]
\Require Koopman matrices $F\in\mathbb{R}^{m\times m}$,
         $G\in\mathbb{R}^{n\times n}$; dictionaries $\Phi$, $\Psi$;
         conjugacy family $\{h_\theta\}_{\theta\in\Theta}$;
         evaluation grid $\{x_i\}_{i=1}^M$;
         regularisation $\lambda>0$;
         coarse grid $\theta_1,\dots,\theta_L\in\Theta$;
         number of refinements $n_{\mathrm{ref}}$
\Ensure  $d_{\mathrm{CSA}}$, $\theta^*$
\Statex
\State $\mathbf\Phi_1 \gets [\Phi(x_1),\dots,\Phi(x_M)]^\top \in\mathbb{R}^{M\times m}$
\Comment{Pre-compute dictionary on grid}
\State $G_\Phi \gets \mathbf\Phi_1^\top \mathbf\Phi_1 + \lambda I_m$
\Comment{Regularised Gram matrix}
\State $R \gets G_\Phi^{-1}\,\mathbf\Phi_1^\top$
\Comment{Regression operator, $m\times M$}
\Statex
\Function{Objective}{$\theta$}
    \State $\mathbf\Phi_2 \gets [\Psi(h_\theta(x_1)),\dots,\Psi(h_\theta(x_M))]^\top$
    \Comment{Dictionary of system 2 at $h_\theta(x)$}
    \State $P_\theta \gets (R\;\mathbf\Phi_2)^\top$
    \Comment{Least-squares transfer matrix, $n\times m$}
    \State \Return $\|P_\theta F - G\,P_\theta\|_F^2 \;/\; \|P_\theta\|_F^2$
\EndFunction
\Statex
\State \textbf{Coarse search:} evaluate \Call{Objective}{$\theta_j$} for $j=1,\dots,L$
\State $\theta^* \gets \arg\min_j$ \Call{Objective}{$\theta_j$}
\Statex
\State \textbf{Local refinement:} for the top-$n_{\mathrm{ref}}$ coarse candidates, run Brent's method (bounded scalar optimisation) on \Call{Objective}{$\cdot$} in a window $[\theta_j - \delta, \theta_j + \delta]$, where $\delta$ is the coarse grid spacing; update $\theta^*$ if any refinement improves the objective
\Statex
\State $d_{\mathrm{CSA}} \gets \sqrt{\textsc{Objective}(\theta^*)}$
\State \Return $d_{\mathrm{CSA}}$, $\theta^*$
\end{algorithmic}
\end{algorithm}

% ─────────────────────────────────────────────────────────
\subsection{CSA-Krylov}

In the delay-embedding setting the dictionaries are determined by the data
through the SVD of the Hankel matrix (Appendix~\ref{app:krylov}). Given the
known dynamics $f_i$ and the SVD basis $V_{i,r}$, the projected Krylov
observables can be evaluated at any point~$x$:
\[
  \phi_i(x) \;=\; V_{i,r}^\top
  \begin{pmatrix} x \\ f_i(x) \\ f_i^2(x) \\ \vdots \\ f_i^{d-1}(x) \end{pmatrix}
  \;\in\;\mathbb{R}^r.
\]
This is the analogue of the EDMD dictionary~$\Phi(x)$: a scalar state is
mapped to an $r$-dimensional feature vector whose entries are the leading
SVD components of the delay-embedding. CSA-Krylov then constructs $P_\theta$
and optimises over $\theta$ exactly as CSA (Algorithm~\ref{alg:CSA}), with
$\phi_1,\phi_2$ replacing $\Phi,\Psi$.

\section{Convergence diagnostics}
\label{sec:convergence}

We test that the Koopman operators we fit converge as the
representation truncation refines.

For the EDMD panels, we sweep $k\in\{2.0, 2.2,\ldots,4.0\}$ and fit two
Koopman matrices per~$k$: $F$ on trajectories of $f_k(x)=k\,x(1-x)$, and
$G$ on trajectories of $g_k=h_\theta\circ f_k\circ h_\theta^{-1}$ with
$h_\theta(x)=x+\theta\sin(\pi x)$ and $\theta=0.10$, matching the EDMD
experiments of Section~\ref{sec:results_edmd}; both EDMD-Fourier and
EDMD-Legendre are swept across $n_b\in\{11, 21, 41, 81, 121\}$ basis
functions.  For the Hankel-DMD panel we replay the smooth-bijection
setup of Section~\ref{sec:CSA_krylov}: linear contractions $f_r$ at
$r\in\{0.5, 0.8\}$ and their conjugate twins $g^{(r)}=h_{a^*,n^*}^{-1}
\circ f_r\circ h_{a^*,n^*}$ at the operating point $(a^*, n^*)=(0.5, 4)$,
with delay dimension $d=100$ fixed and rank swept across
$\{2,3,4,5,6,8,10,15,20,50,80\}$.

Figure~\ref{fig:appx-residual} reports the held-out state-prediction
residual
\[
\rho_x(K, f) \;=\; \frac{\sum_t \big(x_{t+1} - c^{\!\top} K\,\Phi(x_t)\big)^2}
                        {\sum_t x_{t+1}^2},
\qquad
c \;=\; \arg\min_{c}\; \sum_t \big(x_t - c^{\!\top} \Phi(x_t)\big)^2,
\]
where the readout $c$ is fit on training data so that
$c^{\!\top}\Phi(x)$ recovers the state. A natural length-scale is obtained by considering that the entire state-space is given by the unit interval $[0,1]$. Predicting the $1$-D state
with a fixed readout makes $\rho_x$ comparable across truncation
levels.  Fourier-EDMD converges algebraically in $n_b$ with a peak at $k=4$ where the state is most non-periodic.
On the other hand, Legendre-EDMD on $f_k$ has converged at every $n_b$.  For the Hankel-DMD panel the linear contractions saturate
at the floor already at rank $r=2$ — their Koopman spectrum is
essentially two-dimensional — while the smooth-bijection conjugates
require a larger truncation to be resolved: at $r=0.8$ the residual
drops by about four orders of magnitude between $r=2$ and $r=8$, after which it
plateaus.  All three representations exhibit monotone refinement,
confirming that the fitted operators converge under the truncation
sweep.

\begin{figure}[h]
  \centering
  \includegraphics[width=\linewidth]{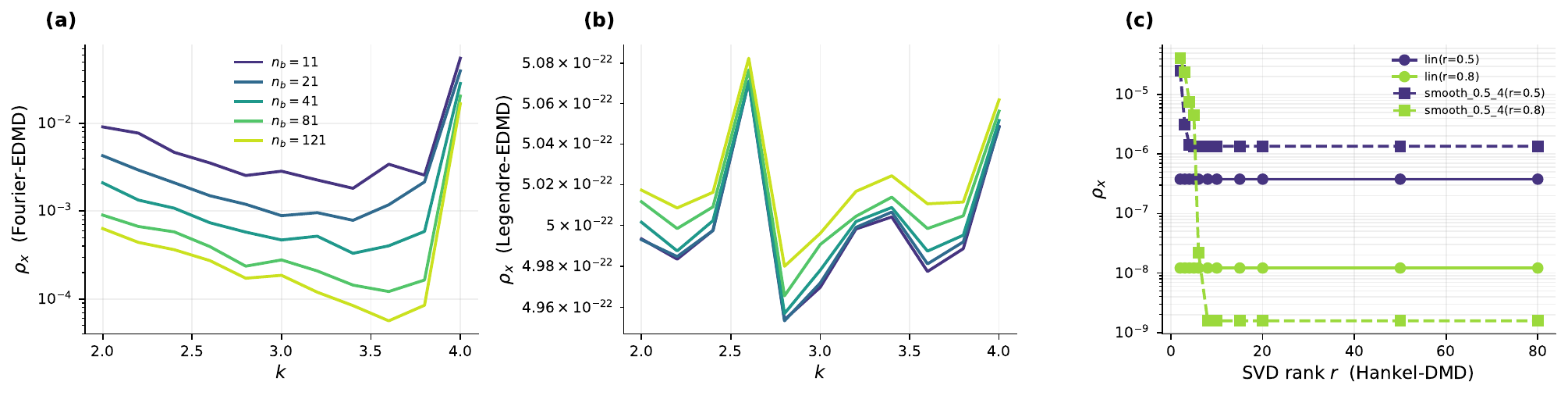}
  \caption{State-prediction residual $\rho_x$ versus truncation, by
    representation.  \textbf{(a)}~Fourier-EDMD on $f_k$; lighter
    shades denote larger $n_b$.  \textbf{(b)}~Legendre-EDMD on $f_k$. \textbf{(c)}~Hankel-DMD on the smooth-bijection setup of
    Section~\ref{sec:CSA_krylov} (linear contractions $f_r$ and
    their smooth-conjugate twins $g^{(r)}$ at $(a^*, n^*) = (0.5,
    4)$); the nonlinear conjugates need a larger truncation than
    the linear bases before $\rho_x$ saturates.}
  \label{fig:appx-residual}
\end{figure}

\newpage

\end{document}